\pdfoutput=1
\documentclass{article}
\usepackage{amsmath,amsthm,amssymb, ascmac,array,mathtools}
\usepackage{bm}

\mathtoolsset{showonlyrefs}
\usepackage{fullpage}
\newtheorem{theorem}{Theorem}
\newtheorem{example}{Example}
\newtheorem{lemma}{Lemma}
\newtheorem{corollary}{Corollary}

\theoremstyle{definition}
\newtheorem{definition}{Definition}

\newcommand{\R}{\mathbb{R}}
\DeclareMathOperator*{\argmin}{arg\,min}
\newcommand{\cM}{\bm{M}}
\newcommand{\cC}{\bm{C}}
\newcommand{\cA}{\bm{A}}

\usepackage[utf8]{inputenc}
\usepackage[sort,numbers]{natbib}
\usepackage[switch]{lineno}
\usepackage[hidelinks]{hyperref}

\usepackage{authblk}
\usepackage{booktabs}

\title{\mbox{Envy-freeness and maximum Nash welfare} \mbox{for mixed divisible and indivisible goods}}
\author[1]{Koichi Nishimura}
\affil{CRESCO LTD.}
\author[2]{Hanna Sumita}
\affil{Tokyo Institute of Technology}
\date{}
\begin{document}
\maketitle

\begin{abstract}
We study fair allocation of resources consisting of both divisible and indivisible goods to agents with additive valuations. 
When only divisible or indivisible goods exist, it is known that an allocation that achieves the maximum Nash welfare (MNW) satisfies the classic fairness notions based on envy.
Moreover, the literature shows the structures and characterizations of MNW allocations when valuations are binary and linear (i.e., divisible goods are homogeneous).
In this paper, we show that when all agents' valuations are binary linear, an MNW allocation for mixed goods satisfies the envy-freeness up to any good for mixed goods (EFXM).
This notion is stronger than an existing one called envy-freeness for mixed goods (EFM), and our result generalizes the existing results for the case when only divisible or indivisible goods exist.
When all agents' valuations are binary over indivisible goods and identical over divisible goods (e.g., the divisible good is money), we extend the known characterization of an MNW allocation for indivisible goods to mixed goods, and also show that an MNW allocation satisfies EFXM.
For the general additive valuations, we also provide a formal proof that an MNW allocation satisfies a weaker notion than EFM.
\end{abstract}

\section{Introduction}

Fair allocation of finite goods has been attracting attention for decades from many research communities such as economics and computer science. 
Most of the literature on fair allocation can be categorized into two types. 
One is the research that deals with heterogeneous and infinitely divisible goods.
In this case, the fair allocation problem is also called cake cutting by merging all goods into one cake. 
The other deals with indivisible goods.
In this paper, we focus on the case where agents' valuations over goods are additive.

A standard fairness notion is \emph{envy-freeness} (EF)~\cite{Foley}.
This property requires that no agent prefers another agent's bundle to her own bundle.
It is well known that for divisible goods, there exists an envy-free allocation such that all goods are distributed~\cite{Varian}, and such an allocation can be found by a discrete and bounded protocol~\cite{Aziz2016}.
On the other hand, for indivisible goods, such an envy-free allocation may not exist.
For example, when there is a single good and there are two agents, either agent can receive the good, and the agent with no good envies the other.
Thus, an approximate fairness notion called \emph{envy-freeness up to one good} (EF1) has been proposed~\cite{Budi11a,Lipton2004}.
An allocation is said to be EF1 if we can eliminate the envy of any agent towards another one by removing some good from the envied agent's bundle.
It is known that an EF1 allocation always exists and can be found in polynomial time~\cite{Lipton2004} under additive valuations.
A similar but stronger notion of EF1 called \emph{envy-freeness up to any good} (EFX) has also been well studied~\cite{Caragiannis2019}.
An allocation is said to be EFX if we can eliminate the envy of any agent towards another one by removing any good from the envied agent's bundle.
It is known that there exists an EFX allocation for special cases~\cite{Amanatidis2021,Chaudhury2020,Mahara2021,Mahara2020,Plaut2020} while the existence is open even for four agents with additive valuations.

Another prominent fairness notion is based on the Nash welfare, which is a geometric mean of all utilities in an allocation.
An allocation is said to be a \emph{maximum Nash welfare} (MNW) allocation if it maximizes the number of agents with positive utilities among all allocations and subject to that, it maximizes the Nash welfare of such agents. 
While maximizing the Nash welfare aims to be socially optimal, it attains other efficiency and fairness among individuals.
By the definition, an MNW allocation satisfies an efficiency notion called \emph{Pareto optimality}~(PO), which means that no agent gain more without hurting others.
As for fairness, it is known that any MNW allocation is EF~\cite{Varian, Segal-Halevi2019} for divisible goods. For indivisible goods, every MNW allocation is EF1~\cite{Caragiannis2019}, and moreover EFX~\cite{Amanatidis2021} when valuations of agents are \emph{binary} (i.e., the value for each good is $0$ or $1$) or bi-valued (i.e., there are at most two possible values for the goods).
When valuations are binary, MNW allocations have an interesting characterization.
Namely, the following three conditions are equivalent for any allocation of indivisible goods~\cite{Benabbou2021,FM2022a}: (i) the Nash welfare is maximized, (ii) the allocation is leximin (i.e., the smallest utility is maximized and subject to that, the second smallest utility is maximized, and so on), and (iii) the allocation is $\Phi$-fair~\cite{hybrid}, which means that the associated utility vector minimizes a \emph{symmetric strictly convex} function $\Phi$ over utilitarian optimal allocations (i.e., one maximizing the total sum of all agents' utilities).
The same equivalence also holds for homogeneous divisible goods~\cite{fujishige1980,Maruyama1978}.

Recently, fair allocation of a mix of divisible goods and indivisible ones has been a prominent research topic since introduced by Bei et al.~\cite{Bei2021}.
See also a survey~\cite{Liu2023survey}.
In a real-world application, such a situation often occurs; for example, an inheritance may contain divisible goods (e.g., land and money) and indivisible goods (e.g., cars and houses).
In an allocation of such mixed goods, neither EF nor EF1 is a reasonable goal; the EF criterion is too strong, but the EF1 is too weak.
Thus, Bei et al.~\cite{Bei2021} introduced a fairness notion called \emph{envy-freeness for mixed goods} (EFM), which is a generalization of EF and EF1.
Intuitively, in an EFM allocation, each agent does not envy any agents with bundles of only indivisible goods for more than one indivisible good; moreover, each agent does not envy agents with bundles having only a fraction of a divisible good.
Contrary to expectation, an MNW allocation for mixed goods is incompatible even with a weaker variant of EFM~\cite{Bei2021}.
Caragiannis et al.~\cite{Caragiannis2019} mentioned that an MNW allocation for mixed goods is \emph{envy-free up to one good for mixed goods} (EF1M), that is, each agent $i$ does not envy another agent $j$ if one indivisible good is removed from the bundle of agent $j$.
Intuitively, when we regard a divisible good as a set of infinitely small indivisible goods, the EF1M criterion is equivalent to the EF1, and thus an MNW allocation satisfies EF1M.
On the other hand, in the EFM allocation, when agent $j$ receives indivisible goods and a fraction of a divisible good, envies from other agents must be eliminated by removing any ``indivisible'' good from agent $j$'s bundle.
This requirement is exactly EFX.

In this paper, we further investigate the relationship between MNW allocations and fairness notions for mixed goods.
In contrast to the general additive case, little is known about the binary valuations.
The binary valuation is a basic tool to model practical situations such as work shift scheduling\footnotemark, and furthermore MNW allocations under binary valuations also have theoretically interesting properties when either divisible or indivisible goods exist.\footnotetext{In the work shift scheduling, a manager would like to assign shifts to agents, who are willing to work (to gain money) but have heterogeneous preferences over the shifts.
A shift is valued as $1$ if the agent prefers to work in the shift, and valued as $0$ otherwise.
A shift that should be fully allocated to one agent corresponds to an indivisible good, and one that can be shared among several agents (and the wage is also divided) corresponds to a divisible good.}
For mixed goods, Kawase et al.~\cite{hybrid} showed the structure of $\Phi$-fair allocations when valuations are binary and \emph{linear}, which means that divisible goods are homogeneous.
They also provided an example that any $\Phi$-fair or leximin allocation is not an MNW one.
Li et al.~\cite{Li2023} presented a truthful mechanism that outputs an allocation satisfying EFM, MNW and leximin criteria, when the valuations are binary over indivisible goods and identical on a single divisible good (e.g., money).

\subsection{Our contribution}
Our contribution is to extend the established connections between MNW allocations and relaxed envy-freeness for either divisible or indivisible goods to the setting of mixed goods.
Specifically, our results are divided into three cases based on the properties of valuation functions.

\paragraph{(i) Binary linear valuations}
One main result is to show that when all agents' valuations are binary and linear, every MNW allocation satisfies a stronger notion than EFM called \emph{envy-freeness up to any good for mixed goods} (EFXM).
Here, we assume that the allocation of divisible goods is either empty or of positive value.
This result generalizes the existing result by Amanatidis et al.~\cite{Amanatidis2021} for indivisible goods.
We remark that the argument regarding a divisible good as a set of infinitely small indivisible goods is not directly applicable for the binary case because valuation for such infinitely small indivisible goods are not binary.
We first prove that $\Phi$-fair allocations are EFXM (Theorem~\ref{thm:phi}) by leveraging the structure of $\Phi$-fair allocations, which is established in~\cite{hybrid}.
If we appropriately set a symmetric strictly convex function $\Phi$, then the MNW criterion coincides with $\Phi$-fairness, which is detailed in Appendix~\ref{sec:MNW}.
These facts imply that MNW allocations are EFXM.

\paragraph{(ii) Binary valuations over indivisible goods and identical valuations over divisible goods}
We show that maximizing Nash welfare is equivalent to leximin and $\Phi$-fair criteria~(Theorem~\ref{thm:phi-equiv mixed goods}) when valuations are binary for indivisible goods and identical for divisible goods.
Our result generalizes the existing results~\cite{Benabbou2021,FM2022a,fujishige1980,Maruyama1978}, and our proof extends that by Benabbou et al.~\cite{Benabbou2021} for indivisible goods.
In our proof, we utilize the detailed structure of $\Phi$-fair allocations for indivisible goods to deal the divisible goods.
In addition, we show that any MNW allocation is EFXM (Theorem~\ref{thm:binary money}), which is not included in the previous case.

\paragraph{(iii) General additive valuations}
When we go beyond binary valuations, the only known fact is that MNW allocations are EF1M.
First, we provide examples (Examples~\ref{ex:bi-valued,identical on C} and~\ref{ex:bi-valued,identical on M}) such that valuations are bi-valued but any MNW allocation is not even weak EFM.
This is in contrast to the known result that any MNW allocation is EFX when only indivisible goods exists~\cite{Amanatidis2021}.
Second, we show that MNW allocations are EFXM for identical valuations (Theorem~\ref{thm:identical}).
Next, for the existing result on EF1M, we describe a simpler proof 
extending the analysis for divisible goods~\cite{Segal-Halevi2019}.

We summarize the connection between MNW allocation and relaxed envy-freeness in Table~\ref{tab:summary}.
We also prove the existence of an allocation that satisfies the relaxed envy-freeness and PO.
Lastly,  in Section~\ref{sec:discussion}, we discuss relationships between our relaxed envy-free notions and other recently proposed fairness notions for mixed goods.

\begin{table}[htbp]
    \centering
    \caption{Guarantee of relaxed envy-freeness for MNW allocations. The symbol $M$ stands for the set of indivisible goods and $C$ for divisible goods.}
    \label{tab:summary}
    \begin{tabular}{c|ccc}
    \toprule
        Valuations & \multicolumn{3}{c}{Goods} \\
        & $M=\emptyset$ & $C=\emptyset$ & mixed \\
        \midrule
        binary linear
        & EF~\cite{Segal-Halevi2019} & EFX~\cite{Amanatidis2021} & EFXM (Theorem~\ref{thm:phi}) \\
        binary ($M$), identical ($C$)
        & EF~\cite{Segal-Halevi2019} & EFX~\cite{Amanatidis2021} & EFXM (Theorem~\ref{thm:binary money}) \\
        identical additive & EF~\cite{Segal-Halevi2019} & EFX~\cite{Plaut2020} & EFXM (Theorem~\ref{thm:identical}) \\
        bi-valued linear & EF~\cite{Segal-Halevi2019} & EFX~\cite{Amanatidis2021} & EF1M~\cite{Caragiannis2019} (Theorem~\ref{thm:EF1M})\\
        general addtive & EF~\cite{Segal-Halevi2019} & EF1~\cite{Caragiannis2019} & EF1M~\cite{Caragiannis2019} (Theorem~\ref{thm:EF1M})\\
        \bottomrule
    \end{tabular}
\end{table}

\subsection{Related work}
There is a vast literature on fair allocation, and the problem of finding an MNW allocation is one of the most studied topics in this area.
For homogeneous divisible goods (linear valuations), finding an MNW allocation is a special case of finding a market equilibrium of the Fisher market model and is formulated as the convex program of Eisenberg and Gale~\cite{Eisenberg1959}~(see, e.g., a textbook~\cite{textbook} for details).
Later, a combinatorial polynomial-time algorithm was proposed~\cite{Devanur2008} and also strongly polynomial-time ones~\cite{Orlin2010, Vegh2016}.
For indivisible goods, the problem is APX-hard even for additive valuations~\cite{Lee2017}.
Thus, approximation algorithms are intensively studied~\cite{Anari2018,Cole2015,Cole2017,Garg2018}.
On the other hand, when valuations are binary, there exist polynomial-time algorithms~\cite{Barman2018,Darman2015}, and also a truthful mechanism that returns a special MNW allocation in polynomial time~\cite{Halpern2020}.

It is worth mentioning here that many of the positive results on binary additive valuations for indivisible goods are generalized to matroid rank function.
It is shown~\cite{Benabbou2021} that the equivalence of MNW, leximin, and $\Phi$-fair allocations holds for the matorid rank functions, and any clean (i.e., removing any item from some agent's bundle decreases the utility of the agent) MNW allocation is EF1.
A truthful mechanism that returns an MNW allocation in polynomial time also exists given (polynomial-time) oracle access to the valuations~\cite{Babaioff2021}.

For mixed goods, recently Kawase et al.~\cite{hybrid} provided a polynomial-time algorithm to find an MNW allocation when valuations are binary, linear, and the divisible goods are \emph{identical}.
Their algorithm is built on the structure of $\Phi$-fair allocations, together with relationships between $\Phi$-fairness and MNW allocations.
They also showed that the problem is NP-hard even when the indivisible goods are identical.
When valuations are binary over indivisible goods and identical on a single divisible good, 
Li et al.~\cite{Li2023} presented a truthful mechanism that outputs an MNW allocation for mixed goods by utilizing the mechanism for indivisible goods~\cite{Halpern2020}.

\section{Preliminaries}\label{sec:preliminaries}
In this section, we describe the fair allocation problem and fairness notions appearing in the present paper. 
For a positive integer $k$, we denote $[k] \coloneqq \{1,2,\dots, k\}$.
Let $N = \{1,2,\dots, n\}$ be the set of $n$ agents.
Let $M$ be the set of $m$ indivisible goods, and $C$ be the set of $\ell$ divisible goods.
We denote by $E=M \cup C$ the set of \emph{mixed} goods.
We often identify the $j$th divisible good $c_j$ with an interval $I_{c_j}=[\frac{j-1}{\ell}, \frac{j}{\ell})$, and regard the entire set of divisible goods as the interval $C=[0,1)$.
A \emph{piece} of $C$ is a Borel set of $C~(=[0,1))$.
Let $\mathcal{C}$ be the family of pieces of $C$.
We denote by $\mu(C')$ the length\footnote{Formally, we define the length by the Lebesgue measure.} of a piece $C'\in \mathcal{C}$.

Each agent $i$ has a valuation function $u_i \colon 2^M \cup \mathcal{C} \to \R_+$, where $\R_+$ is the set of nonnegative real numbers.
For ease of notation, we write $u_i(g)=u_i(\{g\})$ for all indivisible goods $g\in M$.
We assume the following standard properties for \emph{additive} valuation functions.
(i) Each $u_i$ is additive over indivisible goods: $u_i(M')=\sum_{g\in M'} u_i(g)$ for each $M' \subseteq M$.
(ii) The valuation for the empty set is zero: $u_i(\emptyset)=0$ for each $i \in N$.
(iii) The valuation is countably additive over $C$: the valuation for a piece which is a countable union of disjoint subintervals is the sum of the valuation for the subintervals.
(iv) Each $u_i$ is {non-atomic\footnotemark} over $C$: any single point in $C$ has value $0$.
We remark that the property (iv) implies the divisibility: for any piece $C' \in \mathcal{C}$ and any $\alpha \in [0,1]$, there exists a piece $\tilde{C} \subseteq C'$ such that $u_i(\tilde{C}) = \alpha \cdot u_i(C')$. 
\footnotetext{It is often assumed that a valuation $u_i$ over $C$ is represented as an integral of some density function $f_i$, i.e., $u_i(C') = \int_{C'} f_i(x) \mathrm{d}x$. This assumption is slightly stronger than ours.}

We define special classes of valuation functions.
A valuation $u$ is said to be \emph{binary} if $u(g) \in \{0,1\}$ for every good $g \in E$.
We say that valuations $u_1,u_2, \dots, u_n$ are \emph{linear} if $u_i(C')=\sum_{c\in C} u_i(c) \cdot \mu(C'\cap I_{c})$ for any $i\in N$ and $C' \in \mathcal{C}$.
In other words, each divisible good is homogeneous (e.g., a cake with a single flavor) and the agents care only the quantity, while a divisible good is generally heterogeneous (e.g., a cake with multiple flavors or uneven toppings).
The valuations $u_1,u_2,\dots, u_n$ are called \emph{identical} if $u_1=u_2=\dots=u_n$.
When we say that valuations are identical over $C$, it means that $u_1(C') = u_2(C') = \dots =u_n(C')$ for any piece $C' \in \mathcal{C}$.

Throughout the paper, we consider a complete allocation, i.e., allocating all goods.
We assume that all agents have a positive value for some good, and every good receives a positive value from some agent.
If there are goods that are valued as zero by all the agents, we first decide an allocation of the other goods and next allocate the zero-valued goods to the agent with the smallest utility.

An allocation of indivisible goods in $M$ is an ordered partition $\cM=(M_1,\ldots, M_n)$ of $M$ such that each agent $i$ receives the bundle $M_i \subseteq M$.
Similarly, an allocation of divisible goods $C$ is a partition $\cC=(C_1,\ldots,C_n)$ of $C$ such that each component $C_i$ is a piece and each agent $i$ receives $C_i$.
An allocation of the mixed goods $E$ is defined as $\cA=(A_1,\dots, A_n)$ such that $A_i=M_i\cup C_i$ is a bundle for agent $i$.
The \emph{utility} of agent $i$ in an allocation $\cA$ is the sum of the valuations of $M_i$ and $C_i$; for simplicity, we write the utility of agent $i$ as $u_i(A_i)=u_i(M_i)+u_i(C_i)$.

We first define the efficiency notions.
For an allocation $\cA$ for mixed goods, 
another allocation $\cA'$ is said to \emph{Pareto dominate} $\cA$ if it satisfies that $u_i(A'_i) \geq u_i(A_i)$ for all agents $i\in N$ and $u_i(A'_i) > u_i(A_i)$ for some agent $i \in N$.
An allocation $\cA$ is called \emph{Pareto optimal} (PO) if no allocation Pareto dominates $\cA$.
An allocation $\cA$ is said to be \emph{utilitarian optimal} if the total sum of utilities $\sum_{i\in N} u_i(A_i)$ is maximized over allocations.

We consider the following fairness notions in this paper. 
\begin{definition}
An allocation $\cA$ for mixed goods is said to be \emph{envy-free} (EF) if it satisfies $u_i(A_i)\geq u_i(A_j)$ for any agents $i, j\in N$.
\end{definition}
\begin{definition}
An allocation $\cM$ for indivisible goods is said to be \emph{envy-free up to one good} (EF1) if it satisfies for any agents $i, j\in N$, either $u_i(M_i)\geq u_i(M_j)$ or $u_i(M_i) \geq u_i(M_j\setminus \{g\})$ for some $g\in M_j$. 
\end{definition}
\begin{definition}[Bei et al.~\cite{Bei2021}]
An allocation $\cA$ for mixed goods is said to be \emph{envy-free for mixed goods} (EFM) if it satisfies the following condition for any agents $i, j \in N$:
\begin{itemize}
    \item if $C_j = \emptyset$, then $u_i(A_i)\geq u_i(A_j)$
    or $u_i(A_i) \geq u_i(A_j\setminus \{g\})$ for some indivisible good $g\in A_j$,
    \item otherwise (i.e., $C_j\neq \emptyset$), $u_i(A_i)\geq u_i(A_j)$.
\end{itemize}
\end{definition}
\begin{definition}[Bei et al.~\cite{Bei2021}]
An allocation $\cA$ for mixed goods is said to be \emph{weak envy-free for mixed goods} (weak EFM) if it satisfies the following condition for any agents $i, j \in N$:
\begin{itemize}
    \item if $M_j\neq \emptyset$, and additionally either $C_j=\emptyset$ or $u_i(C_j)=0$, then 
    $u_i(A_i) \geq u_i(A_j\setminus \{g\})$ for some indivisible good $g\in A_j$,
    \item otherwise, $u_i(A_i)\geq u_i(A_j)$.
\end{itemize}
\end{definition}
\begin{definition}[Caragiannis et al.~\cite{Caragiannis2019}]
An allocation $\cA$ for mixed goods is said to be \emph{envy-free up to one good for mixed goods} (EF1M) if it satisfies the following condition for any agents $i, j\in N$:
\begin{itemize}
    \item if $M_j=\emptyset$, then $u_i(A_i)\geq u_i(A_j)$,
    \item otherwise, $u_i(A_i) \geq u_i(A_j\setminus \{g\})$ for some indivisible good $g\in A_j$.
\end{itemize}
\end{definition}

Since the EFM is based on partially the EFX criterion, we introduce a stronger variant of EFM.
We remark that while this concept is novel, the idea appears also in~\cite{Bei2021} without definition, and Li et al.~\cite{Li2023} independently introduced the same notion under the name EFM${}_{\geq 0}$ when $C$ consists of a single divisible good.
\begin{definition}
    An allocation $\cA$ for mixed goods is said to be \emph{envy-free up to any good for mixed goods} (EFXM) if it satisfies the following condition for any agents $i, j\in N$:
\begin{itemize}
    \item if $C_j = \emptyset$, then 
    $u_i(A_i)\geq u_i(A_j)$ or $u_i(A_i) \geq u_i(A_j\setminus \{g\})$ for any indivisible good $g\in A_j$.
    \item otherwise, $u_i(A_i)\geq u_i(A_j)$.
\end{itemize}
\end{definition}
By definition, EF1M is weaker than weak EFM.
It is not difficult to see that any EFXM allocation is EFM.
Hence, we can summarize the relationships between notions as
$$
\text{EF} \Rightarrow \text{EFXM} \Rightarrow \text{EFM} \Rightarrow \text{weak EFM} \Rightarrow \text{EF1M}.
$$

It is known that when agents have additive valuations, there always exists an EF allocation for divisible goods~\cite{Varian}, and an EF1 allocation for indivisible goods~\cite{Lipton2004}
The existence of EFX allocations for indivisible goods is partially known.
For mixed goods, Bei et {al.~\cite{Bei2021}\footnotemark} showed that an EFM allocation always exists for additive valuations, and also mentioned that an EFXM allocation exists whenever there exists an EFX allocation of indivisible goods.

\footnotetext{
More specifically, Bei et al.~\cite{Bei2021} assumed that each agent's valuation is represented as an integral of some density function.
Their existence results are built on a perfect allocation for any subsets of agents and any piece, and the existence of such an allocation is guaranteed also for countably additive and non-atomic valuations by the Dubins-Spanier theorem~\cite{DS1961}.
Thus, we can carry their results to our setting.
}

Finally, we introduce fairness notions based on optimizing a function value.
The \emph{Nash welfare} of an allocation $\cA$ is a geometric mean of the agents' utilities, i.e., $\left(\prod_{i\in N} u_i(A_i)\right)^{\frac{1}{n}}$.
Note that maximizing the Nash welfare is equivalent to maximizing the product of utilities.
In general, it may happen that the Nash welfare is $0$ for all allocations. 
To deal with such cases, the following definition is used in the literature.
\begin{definition}
    An allocation $\cA$ is said to be a \emph{maximum Nash welfare} (MNW) allocation if the number of agents with positive utilities is maximized among all allocations, and subject to that, the Nash welfare of those agents (equivalently, $\prod_{i\in N:\, u_i(A_i) >0} u_i(A_i)$) is maximized.
\end{definition}

Intuitively, maximizing Nash welfare aims to make the utilities of agents as even as possible.
In fact, an MNW allocation leads to EF for indivisible goods~\cite{Caragiannis2019} or EF1 for divisible goods~\cite{Segal-Halevi2019} when agents have additive valuations.

In the remainder, we focus on a notion called $\Phi$-fairness.
For an allocation $\cA$, we define the utility vector, denoted by $z(\cA) \in \R^N$, as $z(\cA)_i= u_i(A_i)$~($i\in N$).
We say that a function $\Phi\colon \R^n\to \R$ is \emph{symmetric} if 
\begin{align}
\Phi(z_1,z_2,\dots,z_n)=\Phi(z_{\sigma(1)},z_{\sigma(2)},\dots,z_{\sigma(n)})
\end{align}
for any permutations $\sigma$ of $(1,2,\dots,n)$, and a function $\Phi\colon \R^n \to\R$ is \emph{strictly convex} if 
\begin{align}
\lambda\Phi(z)+(1-\lambda)\Phi(z')> \Phi(\lambda z+(1-\lambda)z')
\end{align}
for any $z,z'\in\R^n$ and $\lambda\in(0,1)$.
\begin{definition}[Kawase et al.~\cite{hybrid}]
    For a symmetric strictly convex function $\Phi$, 
    an allocation $\cA$ for mixed goods is said to be \emph{$\Phi$-fair} if $\cA$ is utilitarian optimal and the utility vector $z(\cA)$ minimizes $\Phi(z)$ over utility vectors of utilitarian optimal allocations\footnote{Since Kawase et al.~\cite{hybrid} deal with only binary linear valuations, $\Phi$-fairness is defined under the assumption that only utilitarian allocations are considered. We detail the definition here.}.
\end{definition}
Since $\Phi$-fairness requires an allocation to be utilitarian optimal, $\Phi$-fairness is unrelated to (relaxed) envy-freeness in general.
However, when valuations are binary and linear, the equivalence between $\Phi$-fairness, maximizing Nash welfare, and leximin is known for indivisible goods~\cite{Benabbou2021,FM2022a} and for (homogeneous) divisible goods~\cite{fujishige1980,Maruyama1978}, which is formally stated in the following theorem.
Thus, in this case, $\Phi$-fairness implies (relaxed) envy-freeness through the existing result on MNW allocations~\cite{Amanatidis2021,Segal-Halevi2019}.
An allocation $\cA$ is said to be \emph{leximin} if $\cA$ lexicographically maximizes the utility vector, i.e., maximizes the minimum utility, subject to that maximizes the second minimum, and so on.
\begin{theorem}[Benabbou et al.~\cite{Benabbou2021}, Frank and Murota~\cite{FM2022a}, Fujishige~\cite{fujishige1980}, Maruyama~\cite{Maruyama1978}]\label{thm:Phi-either type}
    Let $\Phi \colon \R^n \to \R$ be any symmetric strictly convex function.
    Suppose that agents' valuations are binary linear.
    If $M=\emptyset$ or $C=\emptyset$, the following statements are equivalent for any allocation $\cA$:
    \begin{enumerate}
        \item $\cA$ is $\Phi$-fair.
        \item $\cA$ is a leximin allocation.
        \item $\cA$ maximizes Nash welfare.
    \end{enumerate}
\end{theorem}

Let us now see the case with the mixed goods.
Even for binary linear valuations, it is known that the equivalence stated in Theorem~\ref{thm:Phi-either type} does not hold in general, even when divisible goods are identical (i.e., each agent desires either all or none of the divisible goods)~\cite{hybrid}.
However, for any instance with binary linear valuations, there exist a function $\Phi'$, specific to the instance, such that any MNW allocations are $\Phi'$-fair and vice versa.
We provide the construction of the function $\Phi'$ in Appendix~\ref{sec:MNW}.
In addition, it is shown in~\cite{hybrid} that if we set $\hat{\Phi}(z)=\sum_{i\in N}(2n)^{-z_i/\varepsilon}$ with an appropriate small number $\varepsilon$, then an allocation $\cA$ is $\hat{\Phi}$-fair if and only if $\cA$ is leximin.
Thus, we can still view $\Phi$-fairness as a generalization of both leximin and maximizing the Nash welfare.

In the present paper, we show that Theorem~\ref{thm:Phi-either type} for indivisible goods can be extended to the case when valuations are binary over $M$ and identical over $C$.
Here, the valuations are not assumed to be linear, but as long as only the utility vector is relevant, this case is essentially the one when $C$ consists of a single homogeneous good.
\begin{lemma}\label{lem:identical homogeneous}
    Suppose that the agents have an identical valuation over $C$.
    Let $\hat{C}$ be a singleton of a homogeneous divisible good such that all the agents value $\hat{C}$ as $u_1(C) = u_2(C)=\dots=u_n(C)$.
    Then for any vector $x$, there exists an allocation $\cA$ of $M\cup C$ with $z(\cA)=x$ if and only if there exists an allocation $\hat{\cA}$ of $M\cup \hat{C}$ with $z(\hat{\cA})=x$.
\end{lemma}
\begin{proof}
    Recall that the valuations over $C$ satisfies additivity and divisibility.
    For any allocation $\cC'$ of $C$, we can construct an allocation $\tilde{\cC}$ of $\hat{C}$ such that $u_i(\tilde{C}_i)=u_i(C'_i)$ for all $i\in N$.
    The argument for any allocation $\tilde{\cC}$ of $\hat{C}$ is similar.
\end{proof}
Based on Lemma~\ref{lem:identical homogeneous}, we present the following theorem, whose proof is deferred to Appendix~\ref{sec:appendix}.
\begin{theorem}\label{thm:phi-equiv mixed goods}
    Let $\Phi \colon \R^n \to \R$ be any symmetric strictly convex function.
    Suppose that agents' valuations are binary over $M$ and identical over $C$.
    The following statements are equivalent for any allocation $\cA$:
    \begin{enumerate}
        \item $\cA$ is $\Phi$-fair.
        \item $\cA$ is a leximin allocation.
        \item $\cA$ maximizes Nash welfare.
    \end{enumerate}
\end{theorem}

\section{The binary valuations}\label{sec:phi}
In this section, we investigate the relationships between relaxed envy-freeness notions and MNW allocations through $\Phi$-fairness.
Thus, we focus on utilitarian optimal allocations throughout this section.

First, assume that the agents have binary linear valuations.
We may consider only utilitarian optimal allocations such that
\begin{align}\label{eq:allocation}
    C_i \cap I_c =\emptyset \text{ or } u_i(C_i\cap I_c)>0 \quad \text{for each } i \in N \text{ and } c \in C.
\end{align}
This is not restrictive because any utilitarian optimal allocation can be transformed to one satisfying the condition~\eqref{eq:allocation} without changing the utilities of the agents and adding any envy as follows.
Suppose that the condition~\eqref{eq:allocation} is not satisfied for some agent $i$ and divisible good $c$, i.e., $C_i\cap I_c \neq \emptyset$ but $u_i(C_i\cap I_c)=0$. 
If the length of $C_i \cap I_c$ is positive, then $u_i(C_i \cap I_c)>0$ must hold because of the utilitarian optimality of the allocation.
Thus, the length of $C_i \cap I_c$ is zero, and $u_k(C_i \cap I_c)=0$ for any agent $k$.
Hence, we may transfer the piece $C_j \cap I_c$ to an agent $j$ who has a piece in $I_c$ with positive length.
The condition~\eqref{eq:allocation} is assumed just for excluding a meaningless allocation such as the one allocating a single point to some agent.

The main result of this section is to show that $\Phi$-fairness implies EFXM.
\begin{theorem}\label{thm:phi}
Let $\Phi$ be a symmetric strictly convex function.
When all agents' valuations are binary and linear, every $\Phi$-fair allocation satisfying \eqref{eq:allocation} is EFXM.
\end{theorem}

In the proof of the main result, we utilize the structure of $\Phi$-fair allocations shown in~\cite{hybrid}.
When we are given a partition $(N^1,\dots, N^q)$ of agents $N$ into $q$ subsets, we define partitions of $M$ and $C$ as follows.
Let $M^1$ be the set of indivisible goods in $M$ that are desired only by some agents in $N^1$.
For $k=2,\dots, q$, let $M^k$ be the set of goods in $M\setminus(M^1\cup\dots\cup M^{k-1})$ that are desired only by some agents in $N^1\cup \dots \cup N^k$.
We define $C^1, \dots, C^q$ in a similar way.
In other words, we define $M^k$ and $C^k$ ($k=1,\dots, q$) as
\begin{align}
    M^k &= \textstyle\left\{g\in M\setminus\bigcup_{k'=1}^{k-1}M^{k'} \mid u_i(g)=0~(\forall i\in N\setminus \bigcup_{k'=1}^k N^{k'}) \right\}, \label{eq:canonical indivisible}\\
    C^k &= \textstyle\left\{c\in C\setminus\bigcup_{k'=1}^{k-1}C^{k'} \mid u_i(c)=0~(\forall i\in N\setminus \bigcup_{k'=1}^k N^{k'})\right\}. \label{eq:canonical divisible}
\end{align}
The partitions of goods are determined only by the valuations of agents, and independent from the function $\Phi$.
\begin{lemma}[Kawase et al.~\cite{hybrid}]\label{lem:canonical}
There exists a natural number $q$, partitions $(N^1,\dots, N^q)$ of $N$ and integers $\beta_1 > \dots >\beta_q$ such that for any symmetric strictly convex function $\Phi$ and any $\Phi$-fair allocation $\cA$, 
\begin{enumerate}
    \item $\beta_k \geq u_i(A_i) \ge \beta_k-1$ for each agent $i\in N^k$ and $k=1,\dots, q$, and
    \item for $k=1,\dots,q$, only agents in $N^k$ can receive a fraction of each good in $M^k\cup C^k$, where $M^k$ and $C^k$ are defined as \eqref{eq:canonical indivisible} and \eqref{eq:canonical divisible}, respectively.
\end{enumerate}
\end{lemma}
We remark that the partition of agents given in Lemma~\ref{lem:canonical} is called the \emph{canonical partition}, and the integers are called \emph{essential values} in the literature.
See also existing papers~\cite{FM2022a,hybrid} and references therein for the details of these concepts.

We also use the following lemma, which is used to show that MNW implies EF for divisible goods~\cite{Segal-Halevi2019}.
\begin{lemma}[Stromquist and Woodall~\cite{Stromquist1985}]\label{lem:sw}
Let $i,j$ be two agents. 
Let $H \subseteq [0, 1)$ be a piece such that $u_i(H)>0$ and $u_j(H) > 0$. 
Then, for every $\alpha \in [0,1]$, there exists a piece $H^z \subseteq H$ such that $u_i(H^z)=\alpha \cdot u_i(H)$ and $u_j(H^z)=\alpha \cdot u_j(H)$.
\end{lemma}

Now we are ready to proof Theorem~\ref{thm:phi}.
For each $i \in [n]$, we denote by $\chi_{i}$ a unit vector in $\{0,1\}^n$ whose $i$th entry is $1$, and all other entries are $0$.
\begin{proof}[Proof of Theorem~\ref{thm:phi}]
    Let $\cA$ be a $\Phi$-fair allocation.
    First, since $\cA$ is utilitarian optimal over allocations,
    every agent $i$ has a positive valuation for any indivisible good in $M_i$.
    For each agent $i$, let $C^+_i = \{ C_i \cap I_c \mid u_i(C_i \cap I_c)>0, \ c \in C\}$.
    We observe that, for each $i,j\in N$, we have $u_i(g)\leq u_j(g)~(=1)$ for any $g\in M_j$ and $u_i(c') \leq u_j(c')$ for any $c' \in C^+_j$.
    Hence we have
    \begin{align}
        u_i(A_j) \leq u_j(A_j)\quad (i,j\in N). \label{eq:compare}
    \end{align}

    Let $(N^1,\dots,N^q)$ be the partition of agents indicated in Lemma~\ref{lem:canonical}, and let $\beta_1>\dots>\beta_q$ be the corresponding integers. 
    Let $(M^1,\dots,M^q)$ and $(C^1,\dots,C^q)$ be the partitions of $M$ and $C$ defined as \eqref{eq:canonical indivisible} and \eqref{eq:canonical divisible}, respectively.
    Focus on an arbitrary agent $i$ in $N^k$ ($1 \leq k \leq q$).
    By definition, agent $i$ does not desire goods in $M^{k'} \cup C^{k'}$ with $k' < k$.
    Hence, $u_i(A_{i'}) = 0$ holds for each $i'\in N^{k'}$ with $k'<k$ by the property~2 in Lemma~\ref{lem:canonical}.
    Additionally, by~\eqref{eq:compare} and the property~1 in Lemma~\ref{lem:canonical}, for each $i'\in N^{k'}$ with $k'>k$, it holds that $u_i(A_{i'}) \leq u_{i'}(A_{i'}) \leq \beta_{k'} \leq \beta_k-1 \leq u_i(A_i)$.
    Therefore, agent $i$ does not envy agents not in $N^{k}$.

    Suppose that agent $i \in N^k$ envies agent $j \in N^k$, i.e., $u_i(A_i)<u_i(A_j)$.
    We observe that $z(\cA)_i=u_i(A_i) < u_i(A_j) \leq u_j(A_j)=z(\cA)_j$ by \eqref{eq:compare}.

    \paragraph{Case 1} Assume that $u_i(C_j)=0$.
    Then it holds that 
    $$\beta_k-1 \leq u_i(A_i) < u_i(M_j) \leq u_j(M_j) \leq \beta_k,$$
    and hence $u_j(M_j)=\beta_k$ by integrality.
    By Lemma~\ref{lem:canonical} again, we have $u_j(C_j) =0$, which implies that $C_j = \emptyset$ by \eqref{eq:allocation}.
    Then, $M_j$ is nonempty, and suppose that there exists an indivisible good $g\in M_j$ such that $u_i(A_i)<u_i(A_j \setminus \{g\})$.
    \begin{itemize}
        \item If $u_i(g)=0$, then by Lemma~\ref{lem:canonical} and the utilitarian optimality of $\cA$, we have
    \begin{align*}
        \beta_k \geq u_j(A_j)= u_j(M_j) = 1+ u_j(M_j \setminus \{g\}) > 1+u_i(A_i).
    \end{align*}
    However, this implies that $u_i(A_i)<\beta_k-1$ and this is a contradiction.

        \item Assume that $u_i(g)=1$.
    It holds that 
    $$u_j(A_j)-u_i(A_i)\geq u_i(A_j)-u_i(A_i) = 1+ u_i(A_j \setminus \{g\})-u_i(A_i) > 1.$$
    However, since both agents $i,j$ belong to $N^k$, it must hold that $ u_j(A_j)-u_i(A_i) \leq 1$ by Lemma~\ref{lem:canonical}, and this is a contradiction.
    \end{itemize}

    \paragraph{Case 2} Assume that $u_i(C_j)>0$.
    Choose a piece $c'\in C^+_j$ such that $u_i(c')>0$.
    Note that $u_i(c')=u_j(c')>0$ since valuations are binary and linear.
    Let $\varepsilon$ be a number such that $0<\varepsilon< \min\{u_j(A_j)-u_i(A_i), u_j(c')\}$.
    Note that $u_j(A_j)-u_i(A_i) \geq u_i(A_j)-u_i(A_i)>0$.
    By Lemma~\ref{lem:sw} with $H=c'$ and $\alpha=\varepsilon/u_j(c')$, we can choose a piece $\hat{c} \subseteq c'$ such that $u_i(\hat{c})=u_j(\hat{c})=\varepsilon$.
    Let $\cA'$ be an allocation obtained from $\cA$ by transferring the piece $c'$ to agent $i$.
    The utility vector of $\cA'$ is $z(\cA') = z(\cA) - \varepsilon (\chi_j-\chi_i)$.
    The allocation $\cA'$ is also utilitarian optimal.
    Let $\lambda$ be the number such that $\lambda(u_j(A_j)-u_i(A_i))=\varepsilon$.
    Then, since $\lambda \in (0,1)$, we can see that
    \begin{align}\begin{split}
        \Phi(z(\cA))
        &=\lambda\Phi(z(\cA)-(z(\cA)_j-z(\cA)_i)(\chi_j-\chi_i))+(1-\lambda)\Phi(z(\cA))\\
        &>\Phi(\lambda(z(\cA)-(z(\cA)_j-z(\cA)_i)(\chi_j-\chi_i))+(1-\lambda)z(\cA))\\
        &=\Phi(z(\cA)-\varepsilon(\chi_j-\chi_i)),
    \end{split}
    \label{eq:decrease Phi}
    \end{align}
    which contradicts to $\Phi$-fairness of $\cA$.
    
    Summarizing the arguments so far, the only possible case is when $C_j=\emptyset$ and we can eliminate agent $i$'s envy towards agent $j$ by removing any indivisible good from $M_j$.
    Therefore, $\cA$ is EFXM.
\end{proof}
We remark that even without the condition~\eqref{eq:allocation}, we can say that any $\Phi$-fair allocation satisfies a slightly weaker variant of the EFXM: for any $i,j \in N$, agent $i$ uses the EFX criterion if $u_i(C_j)=0$, and the EF criterion otherwise.

It is clear that any utilitarian optimal allocation is PO.
The above theorem implies that an EFXM and PO allocation exists if a $\Phi$-fair allocation always exists.
We observe the existence of a $\Phi$-fair allocation by using the existence result of an MNW allocation for divisible goods.
\begin{theorem}[Segal-Halevi and Sziklai~\cite{Segal-Halevi2019}]\label{thm:exist MNW divisible}
    When all the goods are divisible and the agents have additive valuations,
    there exists an MNW allocation of divisible goods.
\end{theorem}
\begin{lemma}\label{lem:exist Phi-fair}
    Suppose that all agents' valuations are binary and linear.
    For any symmetric strictly convex function $\Phi$, there exists a $\Phi$-fair allocation.
\end{lemma}
\begin{proof}
    For each utilitarian optimal allocation of indivisible goods $\cM$, we say that an allocation $\cA$ of $E$ is \emph{$\Phi$-fair with respect to $\cM$} if the utility vector $z(\cA)$ minimizes $\Phi$ among utilitarian optimal allocations of $E$ such that indivisible goods are allocated according to $\cM$.
    If there exists a $\Phi$-fair allocation with respect to $\cM$ for all utilitarian optimal allocations $\cM$ of $M$, then there also exists a $\Phi$-fair allocation.
    
    Let $\cM$ be an arbitrary utilitarian optimal allocation of indivisible goods.
    Since we do not need to consider agents who value $C$ as $0$ here, we may assume that all agents value $C$ positively.
    We construct the following relaxed instance to determine an allocation of divisible goods.
    The set of goods is $E$ and we assume that all the goods are divisible and homogeneous.
    A profile $(v_1,\dots, v_n)$ of binary linear valuations is defined as follows.
    For each good $g \in M$, we set $v_i(g)=1$ for the unique agent $i$ with $g\in M_i$, and $v_j(g)=0$ for the other agents $j$.
    For each good $c \in C$, we set $v_i(c)=u_i(c)$ for each $i\in N$. 
    We may consider only allocations satisfying~\eqref{eq:allocation} because the valuations are linear.
    Hence, in every utilitarian optimal allocation for this relaxed instance, every good $g$ is fully allocated to the agent $i$ with $g \in M_i$.

    Because valuations are binary linear in the above relaxed instance, $\Phi$-fairness is equivalent to maximizing Nash welfare by Theorem~\ref{thm:Phi-either type}.
    This together with Theorem~\ref{thm:exist MNW divisible} implies that there exists a $\Phi$-fair allocation $\cA^*$ for the relaxed instance.
    By construction, the set of utilitarian optimal allocations in the relaxed instance coincides with that in the original instance.
    Therefore, $\cA^*$ is also $\Phi$-fair with respect to $\cM$ in the original instance.
\end{proof}

\begin{corollary}\label{cor:EFXM+PO}
    When all agents' valuations are binary and linear, there always exists an EFXM and PO allocation for mixed goods.
\end{corollary}

In the remainder of this section, we assume that valuations over $C$ are identical but not necessarily binary.
We show that the same results as Theorem~\ref{thm:phi} and Corollary~\ref{cor:EFXM+PO} hold in this case.
We remark that the condition~\eqref{eq:allocation} is not restrictive also in this case because if $u_i(C_i \cap I_c)=0$ for some $i\in N$ and $c\in C$, then all agents value $C_i \cap I_c$ as $0$.
\begin{theorem}\label{thm:binary money}
    When all the agents' valuation are binary over $M$ and identical over $C$, every $\Phi$-fair allocation satisfying~\eqref{eq:allocation} is EFXM.
\end{theorem}
\begin{proof}
    Let $\cA$ be any $\Phi$-fair allocation.
    By the utilitarian optimality of $\cA$,
    each agent $i$ values every indivisible good $g \in M_i$ as $1$.
    We also have $u_i(C_j) > 0$ for each $i,j\in N$ with $C_j \neq \emptyset$ by~\eqref{eq:allocation}.
    Hence, we observe that $u_i(A_j) \leq u_j(A_j)$ for any $i,j\in N$.
    In the following, suppose that agent $i$ envies agent $j$, i.e., $u_i(A_i)<u_i(A_j)$.

    \paragraph{Case 1} Assume that $u_i(C_j) =0$, which means $C_j =\emptyset$ by~\eqref{eq:allocation}.
    To prove that the EFXM criterion holds, we assume the contrary that there exists some good $g\in M_j$ such that $u_i(A_j \setminus \{g\}) > u_i(A_i)$.
    This implies that $u_j(A_j) - u_i(A_i) = u_j(A_j\setminus \{g\})+1 - u_i(A_i) \geq u_i(A_j \setminus \{g\}) + 1 -u_i(A_i)> 1$.
    Since $u_i(A_i)<u_i(A_j) = u_i(M_j)$, there must be an indivisible good $g^* \in M_j$ with $u_i(g^*)=1$.
    Let $\cA'$ be an allocation obtained from $\cA$ by transferring such a good $g^*$ to agent $i$, and let $\lambda = 1/(u_j(A_j) - u_i(A_i)) \in (0,1)$.
    Here $z(\cA') = z(\cA) - (\chi_j - \chi_i) = \lambda(z(\cA) - (z(\cA)_j-z(\cA)_i)(\chi_j-\chi_i))+(1-\lambda)z(\cA)$.
    Then $\cA'$ is utilitarian optimal but we see that $\Phi(z(\cA)) > \Phi(z(\cA'))$ in a similar way to~\eqref{eq:decrease Phi} in the proof of Theorem~\ref{thm:phi}.
    This contradicts $\Phi$-fairness of $\cA$.

    \paragraph{Case 2} Assume that $u_i(C_j)>0$.
    Let $\varepsilon$ be a number such that $0<\varepsilon< \min\{u_j(A_j)-u_i(A_i), u_j(C_j)\}$.
    Note that $u_j(A_j)-u_i(A_i) \geq u_i(A_j)-u_i(A_i)>0$.
    By divisibility of valuation functions, we can choose a piece $\hat{C} \subseteq C_j$ such that $u_i(\hat{C})=u_j(\hat{C})=\varepsilon$.
    Let $\cA'$ be a utilitarian optimal allocation obtained from $\cA$ by transferring the piece $\hat{C}$ to agent $i$, and let $\lambda = \varepsilon/(u_j(A_j)-u_i(A_i))$.
    Then, we see that $\Phi(z(\cA)) > \Phi(z(\cA'))$ in the same way as~\eqref{eq:decrease Phi}, and this contradicts to the choice of $\cA$.

    Therefore, the only possible case is when $C_j=\emptyset$ and agent $i$'s envy towards agent $j$ is eliminated by removing any indivisible good from $M_j$, and hence $\cA$ is EFXM.
\end{proof}

In a similar way to Corollary~\ref{cor:EFXM+PO}, we can see that there exists an EFXM and PO allocation under the assumption of Theorem~\ref{thm:binary money} if a $\Phi$-fair allocation exists.
By Theorem~\ref{thm:phi-equiv mixed goods}, it suffices to see that an MNW allocation exists. 
Since the Nash welfare is determined by the utility vector, Lemma~\ref{lem:identical homogeneous} implies that we can treat the divisible goods as a single homogeneous one without loss of generality.
Li et al.~\cite{Li2023} show that an MNW allocation always exists when $C$ consists of a single homogeneous divisible good.
Thus, we observe the following corollary.
\begin{corollary}
    When all agents' valuations are binary over $M$ and identical over $C$, there always exists an EFXM and PO allocation for mixed goods.
\end{corollary}

Note that the even EFM and PO are incompatible for general additive valuations, as pointed out by Bei et al.~\cite{Bei2021}.

\section{The general additive valuations}\label{sec:main}
In this section, we explore connections between MNW allocations and relaxed envy-freeness when agents have more general valuations than binary ones.
It is well known that leximin and EF1 are incompatible in general when only indivisible goods exist, and hence leximin and EF1M are also incompatible for mixed goods.
On the other hand, any MNW allocation for indivisible goods still satisfies relaxed envy-freeness: EFX for bi-valued valuations~\cite{Amanatidis2021}, and EF1 for general additive valuations~\cite{Caragiannis2019}.
Thus, we discuss the extensibility of these relationships.

\subsection{Connection to EFXM}
First, we focus on bi-valued valuations.
We say that valuation functions $u_1,\dots, u_n$ are \emph{bi-valued} if there are two distinct nonnegative real numbers $\alpha, \beta$ such that $u_i(e)\in \{\alpha,\beta\}$ for each agent $i\in N$ and good $e\in E$.
The bi-valued valuations in which one value is $0$ are equivalent to binary ones.

In fact, there exists an EFXM allocation for bi-valued additive valuations.
Bei et al.~\cite{Bei2021} implicitly showed that, if the profile of the agents' valuations over indivisble goods belongs to the class such that an EFX allocation exists, then an EFXM allocation always exists.
However, we provide two examples to show that any MNW allocation is not necessarily weak EFM.

\begin{example}[bi-valued, linear, identical over $C$]\label{ex:bi-valued,identical on C}
    Consider an instance with $N=\{1,2\}$, $M=\{a,b\}$, $C=\{c\}$.
    The valuations are $v_1(a)=v_1(b)=1+\varepsilon$, $v_2(a)=v_2(b)=1$, and $v_1(c)=v_2(c)=1$.

    When agent $1$ receives both indivisible goods, the Nash welfare is at most  $2(1+\varepsilon)\cdot 1 = 2(1+\varepsilon)$, which is attained by giving $c$ to agent $2$.
    When agent $2$ receives both indivisible goods, the Nash welfare is at most $1\cdot 2 = 2$, which is attained by giving $c$ to agent $1$.
    When each agent receives one indivisible good, the Nash welfare is maximized by allocating $\frac{1-\varepsilon}{2}$ of $c$ to agent $1$ and the remainder to agent $2$.
    This allocation attains the maximum Nash welfare, which is $((1+\varepsilon)+\frac{1-\varepsilon}{2}) (1+\frac{1+\varepsilon}{2}) = \frac{(3+\varepsilon)^2}{4}$.
    In any MNW allocation $\cA$, since $u_1(A_2) = 1+\varepsilon +\frac{1+\varepsilon}{2} > 1+\varepsilon + \frac{1-\varepsilon}{2} = u_1(A_1)$, agent $1$ envies agent $2$ while $u_1(C_2)>0$.
    This indicates that any MNW allocation is not even weak EFM.
\end{example}
\begin{example}[bi-valued, linear, identical over $M$]\label{ex:bi-valued,identical on M}
    Consider an instance with $N=\{1,2\}$, $M=\{a,b\}$, $C=\{c\}$.
    The valuations are $v_i(g)=1$ for each $i\in N$ and $g\in M$, and $v_1(c)=1+\varepsilon$ and $v_2(c)=1$.

    In a similar way to the previous example, we observe that the Nash welfare is maximized by allocating one indivisible good to each agent, and allocating $\frac{1+2\varepsilon}{2(1+\varepsilon)}$ of $c$ to agent $1$, and the remainder to agent $2$.
    In such an allocation $\cA$, since $u_2(A_2) = 1+\frac{1}{2(1+\varepsilon)} < 1+\frac{1+2\varepsilon}{2(1+\varepsilon)} = u_2(A_1)$, agent $2$ envies agent $1$ while $u_2(C_1)>0$.
    This indicates that any MNW allocation is not even weak EFM in this instance.
\end{example}

Therefore, it remains open whether an EFXM and PO allocation always exists for bi-valued valuations.

Examples~\ref{ex:bi-valued,identical on C} and~\ref{ex:bi-valued,identical on M} suggest that MNW allocations may fail to satisfy EFXM when valuations are identical over at most one of $M$ and $C$.
To complement this fact, we focus on identical valuations over both $M$ and $C$.
\begin{theorem}\label{thm:identical}
    When all the agents' valuation are identical and additive, every MNW allocation satisfying~\eqref{eq:allocation} for mixed goods is EFXM.
\end{theorem}
\begin{proof}
    Let $\cA$ be an MNW allocation. 
    For simplicity of notation, we denote by $u$ the valuation function of agents.
    Suppose that agent $i$ envies agent $j$, i.e., $u(A_i)<u(A_j)$.

    Assume that $C_j \neq \emptyset$, which means that $u(C_j) > 0$ by the condition~\eqref{eq:allocation}.
    Let $\alpha\in (0,1)$ be any number such that $\alpha\cdot u(C_j) < u(A_j)-u(A_i)$.
    By divisibility of $u$, we choose a piece $C' \subseteq C_j$ such that $u(C')= \alpha \cdot u(C_j)$.
    Let $\cA'$ be an allocation obtained from $\cA$ by transferring $C'$ to agent $i$.
    In $\cA'$, the number of agents with positive utilities is increased, or the number of such agents remains the same as $\cA$ but the Nash welfare of them is larger than $\cA$ because  
    \begin{align}
        &(u(A_i)+u(C'))(u(A_j)-u(C'))-u(A_i)u(A_j) \\
        &=u(C') \cdot \left(u(A_j)-(u(A_i)+u(C')) \right) 
        > 0.
    \end{align}
    This contradicts that $\cA$ is an MNW allocation.
    Therefore, $C_j=\emptyset$ must hold, which implies that $M_j \neq \emptyset$.
    
    Supposing that the EFXM condition is not satisfied, we choose an indivisible good $g^*\in M_j$ such that 
    $u(A_i) < u(A_j \setminus \{g^*\})$.
    Let $\cA'$ be an allocation obtained from $\cA$ by transferring $g^*$ to agent $i$.
    In a similar way as above, $\cA'$ improves the Nash welfare, and thus $\cA$ cannot be an MNW allocation.
    Therefore, the EFXM criterion must be satisfied.
\end{proof}

\subsection{Connection to EF1M}
Next, we consider general additive valuations.
Caragiannis et al.~\cite{Caragiannis2019} mentioned that 
MNW allocations are EF1M, which extends the relationship for indivisible goods.
We also provide a simple proof that naturally combines the proofs of existing results.
\begin{theorem}[Caragiannis et al.~\cite{Caragiannis2019}]\label{thm:EF1M}
When all agents have additive valuations, every MNW allocation for mixed goods is EF1M.
\end{theorem}
\begin{proof}
Let $\cA$ be an MNW allocation.
We observe that $\cA$ is PO because if some allocation Pareto dominates $\cA$, then we can improve Nash welfare.
Then if agent $i$ has a positive value for some subset $A'$ of agent $j$'s bundle $A_j$, i.e., $u_i(A')>0$, then $u_j(A')>0$ holds.
Suppose that agent $i\in N$ envies agent $j$.

Assume that $u_i(C_j)>0$, and either $M_j=\emptyset$ or $\frac{u_j(g)}{u_i(g)} > \frac{u_i(C_j)}{u_i(C_j)}~(\forall g\in M_j)$.
Note that this case includes when $M_j \neq \emptyset$ but $u_i(M_j)=0$.
In this case, $u_j(C_j)>0$ also holds.
Choose a number $\alpha\in (0,1)$ such that $\alpha\cdot u_i(C_j) < u_i(A_j)-u_i(A_i)$.
By Lemma~\ref{lem:sw} with $H=C_j$, there exists 
a piece $C' \subseteq C_j$ such that 
$u_i(C')=\alpha\cdot u_i(C_j)$ and $u_j(C')=\alpha\cdot u_j(C_j)$.
Thus, we have $u_i(A_i) + u_i(C') < u_i(A_j)$ and 
$\frac{u_j(C')}{u_i(C')} = \frac{u_j(C_j)}{u_i(C_j)} \leq \frac{u_j(M_j)+u_j(C_j)}{u_i(M_j)+u_i(C_j)} = \frac{u_j(A_j)}{u_i(A_j)}$.
Let $\cA'$ be an allocation obtained from $\cA$ by transferring $C'$ to agent $i$.
In $\cA'$, more agents have positive utilities than $\cA$, or the number of such agents remains the same but the Nash welfare of them is larger than $\cA$ because  
\begin{align}
    &(u_i(A_i)+u_i(C'))(u_j(A_j)-u_j(C'))-u_i(A_i)u_j(A_j) \\
    &=u_i(C') \cdot \left(u_j(A_j)-(u_i(A_i)+u_i(C'))\frac{u_j(C')}{u_i(C')}  \right) \\
    &>u_i(C')\cdot \left(u_j(A_j)-u_i(A_j)\frac{u_j(A_j)}{u_i(A_j)}  \right)
    = 0.
\end{align}
This contradicts that $\cA$ is an MNW allocation.

By the above discussion, it holds that $u_i(C_j)=0$ or $\frac{u_j(g)}{u_i(g)} \leq \frac{u_i(C_j)}{u_i(C_j)}~(\exists g\in M_j)$.
This implies that $u_i(M_j)>0$.
We choose an indivisible good $g^*$ such that $g^* \in \argmin_{g\in M_j:\, u_i(g)>0}\frac{u_j(g)}{u_i(g)}$.
Then we have $\frac{u_j(g^*)}{u_i(g^*)} \leq \frac{u_j(M_j)+u_j(C_j)}{u_i(M_j)+u_i(C_j)} = \frac{u_j(A_j)}{u_i(A_j)}$ if $u_i(C_j)>0$, and $\frac{u_j(g^*)}{u_i(g^*)} \leq \frac{u_j(M_j)}{u_i(M_j)} = \frac{u_j(M_j)}{u_i(A_j)}$ otherwise.
Assume that the EF1M condition is not satisfied.
We have $u_i(A_j \setminus \{g^*\})>u_i(A_i)~(\geq 0)$, and then $u_j(A_j\setminus \{g^*\}) >0$ since $\cA$ is PO.
Let $\cA'$ be an allocation obtained from $\cA$ by transferring $g^*$ to agent $i$.
In a similar way as the former case, we can see that the Nash welfare in $\cA'$ is improved over $\cA$.
Therefore, the EF1M criterion holds.
\end{proof}

We can see that there always exists an MNW allocation of mixed goods in a similar way to the proof of Lemma~\ref{lem:exist Phi-fair}.
This together with Theorem~\ref{thm:EF1M} implies the following.
\begin{corollary}
When all agents have additive valuations, there always exists an EF1M and PO allocation for mixed goods.
\end{corollary}

\subsection{Computation of relaxed envy-free allocations}
In the Robertson-Webb model~\cite{Robertson1998}, which is a standard query model in the literature, it remains open whether an EFM allocation can be found in finite steps.
However, an EF1M allocation can be found in finite steps as follows.

Let $\cC$ be an EF allocation of $C$, which can be found in finite steps~\cite{Aziz2016}.
Let $\cM$ be an EF1 allocation of $M$, which is found in polynomial time~\cite{Lipton2004}.
Let $\cA$ be an allocation defined by $A_i = M_i \cup C_i$ ($i \in N$).
For any agents $i, j$, if $M_j=\emptyset$, then it holds that $u_i(A_i)\geq u_i(C_i) \geq u_i(C_j) = u_i(A_j)$; if $M_j \neq \emptyset$ and $u_i(A_i)<u_i(A_j)$, then $u_i(M_i)< u_i(M_j)$ holds, and there exists an indivisible good $g \in M_j$ such that $u_i(M_i) \geq u_i(M_j\setminus \{g\})$, which implies that
\begin{align*}
    u_i(A_i) = u_i(C_i)+u_i(M_i)
    \geq u_i(C_j)+u_i(M_j \setminus \{g\}) = u_i(A_j \setminus \{g\}).
\end{align*}
Therefore, $\cA$ is EF1M.

We remark that we can also find an EFXM allocation for the binary and linear valuations by combining existing results.
Bei et al.~\cite{Bei2021} mentioned that we can find an EFXM allocation by using their EFM algorithm whenever we have an algorithm to compute an EFX allocation for indivisible goods, and this can be done for binary valuations~\cite{Barman2018}.
Moreover, they also described that their EFM algorithm can run in polynomial time when valuations are linear.
Therefore, we can compute an EFXM allocation in polynomial time.

\section{Discussion}\label{sec:discussion}

Very recently, Li et al.~\cite{Li2024} introduced relaxed fairness notions based an indivisibility ratio.
The \emph{indivisibility ratio} of agent $i \in N$ is defined as $\alpha_i = \frac{u_i(M)}{u_i(M)+u_i(C)}$.
An allocation $\cA$ is said to be \emph{envy-free up to $\alpha$-fraction of one good} (EF$\alpha$) if for any agents $i, j\in N$, either $u_i(A_i) \geq u_i(A_j)$ or $u_i(A_i) \geq u_i(A_j) - \alpha_i \cdot u_i(g)$ for some $g \in M_j$.
This coincides with the EF criterion when $M=\emptyset$, and EF1 when $C=\emptyset$.
Li et al.~\cite{Li2024} showed that an EF$\alpha$ and PO allocation always exists when $n=2$, while no EF$\alpha$ allocation exists when $n\geq 3$.
We note that an MNW allocations is not necessarily EF$\alpha$ as we can see in the following example.
A characterization of MNW allocation satisfying EF$\alpha$ is open.
\begin{example}[MNW may not imply EF$\alpha$]
    Consider an instance where $N=\{1,2\}$, $M=\{a,b\}$ and $C=\{c\}$.
    The valuations are defined as $u_i(g)=1$ for $i\in N$ and $g\in M$, and $u_1(c)=1$ and $u_2(c)=0$.
    Then $\alpha_1 = \frac{2}{3}$ and $\alpha_2 = 1$.
    An allocation $\cA=(\{c\},\{a,b\})$ maximizes the Nash welfare, but $\cA$ is not EF$\alpha$ because $u_1(A_1)=1 < u_1(A_2)-\alpha_1 u_1(g) = \frac{4}{3}$ for any $g\in M_2$.
\end{example}

Li et al.~\cite{Li2024} also showed that any MNW allocation for mixed goods satisfies a relaxed proportionality called PROP$\alpha$: an allocation $\cA$ is said to satisfy \emph{proportionality up to $\alpha$-fraction of one good} (PROP$\alpha$) if for any agent $i\in N$, there exists an indivisible good
$g\in M\setminus M_i$ such that $u_i(A_i) + \alpha_i u_i(g) \geq (u_i(M)+u_i(C))/n$, where 
$\alpha_i = \frac{u_i(M)}{u_i(M)+u_i(C)}$ is a parameter called indivisibility ratio.
This notion generalizes the proportionality (PROP) and the proportionality up to one good (PROP1)\footnote{An allocation $\cA$ is said to satisfy PROP if $u_i(A_i)\geq u_i(E)/n$ for any $i\in N$, and satisfy PROP1 if for any $i\in N$, $u_i(A_i)\geq u_i(E)/n$ or $u_i(A_i) + u_i(g)\geq u_i(E)/n$ for some $g \in M\setminus A_i$}.
It is known that EF (EF1) implies PROP (PROP1) by definition, and EFM implies PROP$\alpha$~\cite{Li2024}.
On the other hand, we show that PROP$\alpha$ and EF1M do not imply each other when both $M$ and $C$ are nonempty.
In the following examples, let $\varepsilon$ be a sufficiently small positive number.
\begin{example}[PROP$\alpha$ but not EF1M]
    Consider an instance where $N=\{1,2,3\}$, $M=\{a,b\}$ and $C=\{c\}$.
    The agents have an identical valuation $u$ such that $u(a)=1-2\varepsilon$, $u(b)=\varepsilon$, and $u(c)=\varepsilon$.
    Then $u(M)+u(C)=1$ and $\alpha_i = 1-\varepsilon$ ($i=1,2,3$).
    An allocation $\cA=(\{a,b\},\{c\},\emptyset)$ is PROP$\alpha$ because $u(A_1) \geq \frac{1}{3}$ and $u(A_i) + \alpha_i u(a) \geq (1-\varepsilon)(1-2\varepsilon) \geq \frac{1}{3}$ for $i=2,3$.
    On the other hand, $\cA$ is not EF1M because the envy of agent $3$ toward agent $1$ cannot be eliminated by removing one good from agent $1$'s bundle.
\end{example}
\begin{example}[EF1M but not PROP$\alpha$]
    Consider an instance where $N=\{1,2\}$, $M=\{a\}$ and $C=\{c\}$, where $c$ is homogeneous.
    The agents have an identical valuation $u$ such that $u(a)=\varepsilon$, and $u(c)=1-\varepsilon$.
    Then $u(M)+u(C)=1$ and $\alpha_1=\alpha_2 = \varepsilon$.
    Let $\cA=(\{a,[0,1/2)\}, \{[1/2,1)\})$.
    This allocation is EF1M because $u(A_1)>u(A_2)$ and the envy of agent $2$ toward agent $1$ can be eliminated by removing the indivisible good $a$.
    On the other hand, $\cA$ is not PROP$\alpha$ because $u(A_2) + \alpha_2 u(a) = \frac{1-\varepsilon + 2\varepsilon^2}{2} < \frac{1}{2}$.
\end{example}
Therefore, the results for PROP$\alpha$ and EF1M have no direct relationship.
Further investigation into relaxations of envy-freeness and proportionality and their relationships for mixed goods is a potential direction for future work.

\section*{Acknowledgments.}
We would like to thank Yasushi Kawase for the comments to improve this paper.
The second author was supported by JSPS KAKENHI Grant Numbers JP17K12646, JP21K17708, and JP21H03397, and JST ERATO Grant Number JPMJER2301, Japan.

\bibliographystyle{abbrvnat}
\bibliography{ref}

\appendix

\section{Relationship of MNW allocations and \texorpdfstring{$\Phi$-fairness}{Φ-fairness}}\label{sec:MNW}

In this section, we describe that for binary linear valuations, there exists a symmetric strictly convex function $\Phi$ such that an utilitarian optimal allocation is $\Phi$-fair if and only if it is an MNW allocation.
Our construction is simpler than that of Kawase et al.~\cite{hybrid}.
Let us fix the fair allocation instance.
We define a function $f \colon \R \to \R$ by
\begin{align}
    f(x) = \begin{cases}
        -\log(x) & (x \geq \frac{1}{(n|E|)^n}) \\
        x^2 - ((n|E|)^n +\frac{2}{(n|E|)^{n}})x + n\log(n|E|) + \frac{1}{(n|E|)^{2n}}+1 & (x < \frac{1}{(n|E|)^n}). 
    \end{cases}
\end{align}
We can see that $f$ is strictly convex function.
Then we set $\Phi(z)=\sum_{i\in N} f(z_i)$, which is a symmetric function.

Let us see a necessary condition for $\Phi$-fair allocations.
By the following lemma, we only need to focus on utility vectors whose elements are zero or at least $1/n$.
\begin{lemma}[Kawase et al.~\cite{hybrid}]\label{lem:lower}
    In every $\Phi$-fair allocation for binary linear valuations, each agent's utility is zero or a positive rational value at least $1/n$.
    The same property holds for every MNW allocation.
\end{lemma}
Let $x$ and $y$ be two such vectors.
Assume that $x$ has $s$ zero elements, $y$ has $t$, and $t < s$.
By a simple calculation, it holds that
\begin{align*}
    \Phi(y)-\Phi(x) &= \sum_{i\in N: y_i\geq 1/n} f(y_i) - \sum_{i\in N: x_i \geq 1/n} f(x_i) +\sum_{i\in N: y_i=0} f(0) - \sum_{i\in N: x_i=0} f(0) \\
    &\leq -(n-t)\log\frac{1}{n} + (n-s)\log|E| -(s-t)\left(n\log(n|E|) + \frac{1}{(n|E|)^{2n}}+1\right)\\
    &= \log (n^{n-t-(s-t)n} |E|^{n-s-(s-t)n} ) -(s-t)\left( \frac{1}{(n|E|)^{2n}}+1\right)\\
    &< \log (|E|^{-1} ) -\left( \frac{1}{(n|E|)^{2n}}+1\right) < 0,
\end{align*}
where the first inequality holds by $x_i \leq |E|$ ($\forall i\in N$), and the last one holds by $s-t\geq 1$.
Therefore, the value of $\Phi$ is minimized at allocations with the smallest number $s^*$ of agents having zero utilities.
Let $\mathcal{S}$ be the set of utilitarian optimal allocations that satisfy the necessary condition in Lemma~\ref{lem:lower} and have exactly $s^*$ agents with zero utility.
All MNW allocations and $\Phi$-fair allocations must belong to $\mathcal{S}$.

By the setting of $\Phi$, we have $\Phi(z(\cA)) = - \sum_{i\in N: u_i(A_i) >0} \log (u_i(A_i)) + s^* \cdot f(0)$ for any allocation $\cA \in \mathcal{S}$.
This implies that minimizing $\Phi$ among allocations in $\mathcal{S}$ is exactly maximizing Nash welfare.
Therefore, an allocation $\cA$ is $\Phi$-fair if and only if $\cA$ is an MNW allocation.

\section{Proof of Theorem~\ref{thm:phi-equiv mixed goods}}\label{sec:appendix}

In this section, we prove Theorem~\ref{thm:phi-equiv mixed goods}.
Throughout this section, we assume that the agents' valuations are binary over $M$ and identical over $C$, and also both $M$ and $C$ are nonempty.
Since valuations over $C$ are identical, we say that a piece is positive if the valuation for the piece is positive.

Let $\hat{C}$ be a singleton of a homogeneous divisible good such that all the agents value $\hat{C}$ as $u_1(C) = u_2(C)=\dots=u_n(C)$.
Lemma~\ref{lem:identical homogeneous} implies that the set of utility vectors for $M\cup C$ coincides with that for $M\cup \hat{C}$.
Since $\Phi$-fairness is determined only by a utlity vector, if an allocation $\cA$ of $M\cup C$ is $\Phi$-fair, then the corresponding allocation $\hat{\cA}$ of $M\cup \hat{C}$, which has the same utility vector, is $\Phi$-fair, and vice versa.
The same relationship holds for leximin and MNW allocations.
Therefore, we assume that $C$ consists of a single divisible good and valuations over $C$ are linear in the remainder of this section.

A key tool in the proof is the \emph{water-filling} procedure.
For a fixed allocation $\cM$ of indivisible goods, the water-filling procedure sequentially allocates the divisible good to the agents with the smallest utility at an equal rate until the divisible good is fully allocated, or until the utility of the agents who receives a nonempty piece reaches the second smallest utility.
If the latter case occurs, the set of agents with the smallest utility is enlarged, and the procedure proceeds with allocation in the same way.
See also~\cite{Li2023} for the detailed description.

Let $\cA$ be an allocation obtained by the water-filling procedure to an allocation $\cM$ of indivisible goods.
We denote by $D$ the set of agents who receive a positive pieces.
Then it is not difficult to see the following holds:
\begin{enumerate}
    \item $z(\cA)_i = z(\cM)_i \geq z(\cA)_j > z(\cM)_j$ for any $i\in N\setminus D$ and $j\in D$, and 
    \item $z(\cA)_i=z(\cA)_j$ for any $i,j \in D$.
\end{enumerate}
We call these properties the \emph{water-filling property}.

Li et al.~\cite{Li2023} show that a leximin (and also MNW) allocation can be found by the water-filling procedure.
\begin{lemma}[Li et al.~\cite{Li2023}]\label{lem:leximin mixed goods}
    Suppose that agents's valuations are binary over $M$ and identical over $C$. 
    There exists a leximin allocation $\cA$ of mixed goods such that $\cA$ is obtained by applying the water-filling procedure to a leximin allocation $\cM$ of indivisible goods.
\end{lemma}

We will show Theorem~\ref{thm:phi-equiv mixed goods} by using this fact.
For a vector $x$, we denote by $x^\uparrow$ the vector obtained from $x$ by rearranging its components in the increasing order.
We call two vectors $x,y\in\R^n$ \emph{value-equivalent} if $x^\uparrow=y^\uparrow$.
We first show that an allocation that is not leximin can be modified so that an agent with smaller utility can receive a piece or an indivisible good from an agent with larger utility.
This result is an extension of the one for indivisible goods by Benabbou et al.~\cite{Benabbou2021}.
\begin{lemma}\label{lem:decrease phi}
    Suppose that agents's valuations are binary over $M$ and identical over $C$. 
    Let $\cA$ be any utilitarian optimal allocation that is not leximin. 
    Then there is another utilitarian optimal allocation $\hat{\cA}$ such that
    either 
    \begin{enumerate}
        \item $z(\hat{\cA})^\uparrow=z(\cA)^\uparrow+\varepsilon(\chi_i-\chi_j)$ for some $\varepsilon \in (0,1)$ and $i,j\in [n]$ such that $(z(\cA)^\uparrow)_j > (z(\cA)^\uparrow)_i + \varepsilon$; or 
        \item $z(\hat{\cA})^\uparrow=z(\cA)^\uparrow+\chi_i-\chi_j$ for some $i,j \in [n]$ with $(z(\cA)^\uparrow)_j > (z(\cA)^\uparrow)_i + 1$.
    \end{enumerate}
\end{lemma}
\begin{proof}    
    Let $\cA^*=(M^*_1\cup C^*_1, \dots, M^*_n\cup C^*_n)$ be a leximin allocation indicated in Lemma~\ref{lem:leximin mixed goods} and 
    let $\cM^*$ of indivisible goods in $\cA^*$.
    Note that $\cA^*$ satisfies the water-filling property, and $C^*_i \neq \emptyset$ if and only if $u_i(C^*_i)>0$ for each $i\in N$.
    
    Let $\cA'=(M'_1\cup C'_1, \dots, M'_n\cup C'_n)$ be an allocation that minimizes 
    \begin{align}\label{eq:symmetric}
        \sum_{i\in N}|M''_i \triangle M^*_i|
    \end{align}
    over the utilitarian optimal allocations $\cA''$ with $z(\cA'')^\uparrow=z(\cA)^\uparrow$.
    Without loss of generality, we let $\cA'$ satisfy~\eqref{eq:allocation}, i.e., $C'_i \neq \emptyset$ if and only if $u_i(C'_i) > 0$ for any $i\in N$.
    In the remainder of the proof, we may also assume that $\cA'$ satisfies the water-filling property.
    If not, there exists agents $i,j$ such that $u_j(C'_j) > 0$ but $u_j(A'_j) > u_i(A'_i)$.
    In this case, setting $\varepsilon \in (0,\min\{1, u_j(A'_j)-u_i(A'_i), u_j(C'_j)\})$, we can obtain another utilitarian optimal allocation $\cA''$ from $\cA'$ by transferring a piece in $C'_j$ with value $\varepsilon$ to agent $i$.
    The existence of such a piece follows from the divisibility.
    Thus, the first statement of the lemma applies.

    By rearranging the indices, we assume that $u_1(A'_1) \leq u_2(A'_2) \leq \dots \leq u_n(A'_n)$, and that $u_i(M^*_i) \leq u_j(M^*_j)$ for any two agents $i < j$ with $u_i(A'_i)=u_j(A'_j)$.
    Correspondingly, we set $u_{\ell_1}(A^*_{\ell_1}) \leq u_{\ell_2}(A^*_{\ell_2}) \leq \dots \leq u_{\ell_n}(A^*_{\ell_n})$.

    Let $k$ be the minimum index such that $u_k(A'_k) < u_{\ell_k}(A^*_{\ell_k})$.
    Since $\cA^*$ is leximin, $u_h(A'_h) = u_{\ell_h}(A^*_{\ell_h})$ for $h \in [k-1]$.
    If $u_h(A'_h) \geq u_h(A^*_h)$ for all $h\in [k]$, then the $k$th smallest value in $z(\cA')$, i.e., $u_k(A'_k)$, is at least that in $z(\cA^*)$, i.e., $u_{\ell_k}(A^*_{\ell_k})$, and this contradicts the choice of $k$.
    Thus, we choose the minimum index $i \in [k]$ such that
    $$u_i(A'_i) < u_i(A^*_i).$$
    Note that for each $h \in [i-1]$, we have 
    $u_{\ell_h}(A^*_{\ell_h}) = u_h(A'_h) \geq u_h(A^*_h)$.
    We see inductively that
    $$u_{\ell_h}(A^*_{\ell_h}) = u_h(A'_h) = u_h(A^*_h)$$
    for each $h \in [i-1]$.

    It holds that $u_i(M'_i) +u_i(C'_i) < u_i(M^*_i)+u_i(C^*_i)$.
    We have the following cases.

    \medskip

    \noindent
    \textbf{Case 1, $u_i(C'_i) < u_i(C^*_i)$.}
    In this case, $C^*_i$ is nonempty.
    We denote by $D'$ the set of agents who receives a positive piece under $\cA'$, i.e., $D'=\{j \in N \mid u_j(C'_j)>0\}$.
    By the water-filling property of $\cA'$, we see that $u_i(A'_i) \geq u_j(A'_j)$ for any $j\in D'$.
    Since $u_i(C^*_i)>0$, agent $i$ has the smallest utility in $\cA^*$ by the water-filling property.
    Then $u_j(A'_j) \leq u_i(A'_i) < u_{i}(A^*_{i}) \leq u_j(A^*_j)$ holds for any $j\in D'$, and we have
    \begin{align*}
        \sum_{j\in D'} u_j(M^*_j) -\sum_{j\in D'} u_j(M'_j) & >  \sum_{j\in D'} u_j(C'_j) -\sum_{j\in D'} u_j(C^*_j) \geq 0,
    \end{align*}
    where the last inequality holds because $C$ is completely distributed among agents in $D'$ under $\cA'$.
    Because the valuations over $M$ are binary and $\cA'$ is utilitarian optimal, there exists an indivisible good $g \in (\bigcup_{j\in D'} M^*_j) \setminus (\bigcup_{j\in D'} M'_j)$.
    Let $j^*$ and $j'$ be the agents who have $g$ under $\cA^*$ and $\cA'$, respectively.
    Note that $j^* \in D'$ and $j' \notin D'$.
    Here, by the water-filling property of $\cA'$, it holds that $u_{j^*}(A'_{j^*}) \leq u_{j'}(A'_{j'})$.

    \noindent
    \textbf{Case 1-1.}
    Suppose that $u_{j^*}(A'_{j^*}) + 1 \geq u_{j'}(A'_{j'})$.
    We show that this case leads to a contradiction.
    \begin{itemize}
        \item If $u_{j^*}(C'_{j^*})\geq 1$, then we transfer $g$ to agent $j^*$ and one unit value in $C'_{j^*}$ to agent $j'$. This decreases the value of~\eqref{eq:symmetric} without changing the utility vector, which contradicts the choice of $\cA'$.

        \item If $u_{j^*}(C'_{j^*})<1$, then $u_{j^*}(A'_{j^*}) = u_{j'}(A'_{j'}) - 1+ u_{j^*}(C'_{j^*})$ because $u_{j'}(A'_{j'})$ is integral and $0<u_{j^*}(C'_{j^*}) \leq u_{j'}(A'_{j'}) - u_{j^*}(M'_{j^*}) \leq 1+ u_{j^*}(C'_{j^*}) < 2$.
        Thus, we transfer $g$ to agent $j^*$ and $C'_{j^*}$ to agent $j'$. 
        The utility of agent $j^*$ is changed to $u_{j^*}(A'_{j^*})-u_{j^*}(C'_{j^*})+1 = u_{j'}(A'_{j'})$, and that of agent $j'$ is changed to $u_{j'}(A'_{j'}) -1 +u_{j'}(C'_{j^*}) = u_{j'}(A'_{j'}) -1 +u_{j^*}(C'_{j^*}) = u_{j^*}(A'_{j^*})$.
        This decreases the value of~\eqref{eq:symmetric} without changing the utility vector up to value-equivalence, which contradicts the choice of $\cA'$.
    \end{itemize}
    
    \noindent
    \textbf{Case 1-2.}
    The remaining case is when $u_{j^*}(A'_{j^*}) + 1 < u_{j'}(A'_{j'})$.
    In this case, we can obtain another utilitarian optimal allocation $\cA''$ such that $z(\cA'') = z(\cA') + \chi_{j^*} - \chi_{j'}$ by transferring $g$ to agent $j^*$, and hence the second statement of the lemma applies.

    \medskip

    \noindent
    \textbf{Case 2, $u_i(C'_i) \geq u_i(C^*_i)$.}
    In this case, $u_i(M'_i) < u_i(M^*_i)$, and there exists $g_1\in M^*_i \setminus M'_i$.
    Let $i_1$ be the agent who has $g_1$ under $\cA'$, i.e., $g_1\in M'_{i_1}$.

    By the choice of $\cA^*$, the allocation $\cM^*$ is leximin for indivisible goods.
    Combining Theorem~\ref{thm:Phi-either type} and Lemma~\ref{lem:canonical} for \emph{indivisible goods}, we can obtain a partition $(N^1,\dots,N^q)$ of agents, integers $\beta_1 > \dots>\beta_q$, and a partition $(M^1,\dots,M^q)$ of $M$ (defined in~\eqref{eq:canonical indivisible}) such that
    \begin{itemize}
        \item $\beta_k \geq u_j(M^*_j) \ge \beta_k-1$ for each agent $j\in N^k$ and $k \in [q]$, and
        \item for $k \in [q]$, only agents in $N^k$ receive goods in $M^k$.
    \end{itemize}
    For each agent $i' \in N$, we denote by $s(i')$ the index in $[q]$ such that $i' \in N^{s(i')}$.

    \begin{description}
        \item[Case 2-1.] If $u_{i_1}(A'_{i_1}) >u_i(A'_i)+1$, then the second statement of the lemma applies in a similar way to Case 1-2.
        \item[Case 2-2.] If $u_{i_1}(A'_{i_1}) = u_i(A'_i)+1$, then without changing the utility vector up to value-equivalence, we can decrease the value of~\eqref{eq:symmetric} by transferring $g_1$ to agent $i$, which contradicts the choice of $\cA'$.
    \end{description}

    \noindent
    \textbf{Case 2-3.}
    Suppose that $u_{i_1}(A'_{i_1}) < u_i(A'_i)+1$.
    We will show that $u_{i_1}(M'_{i_1}) \leq u_{i_1}(M^*_{i_1})$.
    
    Since $g_1$ is allocated to agent $i$ under $\cA^*$ and $i\in N^{s(i)}$, we see that $g_1\in M^{s(i)}$.
    By construction of $M^{s(i)}$, agents in the lower groups $N^{s(i)+1}, \dots, N^q$ value $g_1$ as $0$.
    Since $\cA'$ is utilitarian optimal, $g_1\in M'_{i_1}$ implies that $i_1$ also values $g_1$ as $1$.
    Hence, 
    \begin{align}\label{eq:phi equiv 2}
        i_1 \in N^1 \cup \dots \cup N^{s(i)}.
    \end{align}

    \noindent
    \textbf{Case 2-3-1.}
    First, we further assume that $u_i(C'_i)> 0$.
    By the water-filling property of $\cA'$, it holds that $u_{i_1}(A'_{i_1}) \geq u_i(A'_i)$.
    \begin{itemize}
        \item If $u_i(C'_i) \geq 1$, or $u_i(C'_i) < 1$ and $u_{i_1}(A'_{i_1})$ is integral, then we can decrease the value of~\eqref{eq:symmetric} without changing the utility vector (up to value-equivalence) in the same way as Case 1-1.
        \item The possible case is when $u_i(C'_i) < 1$ and $u_{i_1}(A'_{i_1})$ is not integral. 
        In this case, $u_{i_1}(C'_{i_1})>0$ and hence we have $u_{i_1}(A'_{i_1}) = u_i(A'_i)$ by the water-filling property of $\cA'$.
        By~\eqref{eq:phi equiv 2}, we see that $u_{i_1}(M^*_{i_1}) \geq u_i(M^*_i)$ or $u_i(M^*_i) = u_{i_1}(M^*_{i_1}) + 1$.
        We show that 
        $$u_{i_1}(M'_{i_1}) \leq u_{i_1}(M^*_{i_1})$$
        in either of the following two cases.
        \begin{itemize}
            \item Case $i>i_1$. We have $u_i(A^*_i)>u_i(A'_i) = u_{i_1}(A'_{i_1}) = u_{i_1}(A^*_{i_1})$ by the choice of $i$.
            By the water-filling property of $\cA^*$, we observe that $i_1 \in N^{s(i)}$, $C^*_i = \emptyset$, $u_i(M^*_i) = \beta_{s(i)} = u_{i_1}(M^*_{i_1}) + 1$, and $u_{i_1}(C^*_{i_1}) <1$.
            Thus, since $u_{i_1}(M'_{i_1}) < u_i(A_i^*)$,
            it holds that $u_{i_1}(M'_{i_1}) \leq \beta_{s(i)}-1 = u_{i_1}(M^*_{i_1})$.

            \item Case $i<i_1$. By the ordering and $u_{i_1}(A'_{i_1}) = u_i(A'_i)$, we have $u_i(M^*_i) \leq u_{i_1}(M^*_{i_1})$.
            We also have $u_{i_1}(M'_{i_1}) \leq u_i(M'_i)$ because $u_{i_1}(A'_{i_1}) = u_i(A'_i)$, $u_{i_1}(C'_{i_1})>0$, and $u_i(C'_i)<1$.
            Recall that $u_i(M'_i) < u_i(M^*_i)$ by the assumption of Case 2.
            Combining these, we see that $u_{i_1}(M'_{i_1}) \leq u_{i_1}(M^*_{i_1})$.
        \end{itemize}
    \end{itemize}

    \noindent
    \textbf{Case 2-3-2.}
    Next, we assume that $u_i(C'_i) = 0$, which implies that $u_i(C^*_i)=0$.

    By construction, in $\cA^*$, there exists $\tilde{q}\in [q]$ such that all the agents in $N^{\tilde{q}+1} \cup \dots \cup N^q$ receive positive pieces, and those in $N^1\cup \dots \cup N^{\tilde{q}-1}$ receive an empty piece.
    Since $C^*_i=\emptyset$, we have $s(i) \leq \tilde{q}$.
    If $s(i) < \tilde{q}$, then $C^*_{i_1}=\emptyset$ because of~\eqref{eq:phi equiv 2}.
    Thus, the case $u_{i_1}(C^*_{i_1})>0$ occurs only when $i, i_1 \in N^{\tilde{q}}$.
    However, we show that the case $u_{i_1}(C^*_{i_1})>0$ cannot happen.
    
    Suppose that $u_{i_1}(C^*_{i_1})>0$ (and hence $i, i_1 \in N^{\tilde{q}}$).
    Each agent in $N^{\tilde{q}}$ has $\beta_{\tilde{q}}$ or $\beta_{\tilde{q}}-1$ indivisible goods under $\cM^*$. 
    By the water-filling property, the agents with $\beta_{\tilde{q}}-1$ indivisible goods cannot have more utility than the others in $N^{\tilde{q}}$ under $\cA^*$.
    Since $C^*_i=\emptyset$ but $C^*_{i_1}\neq \emptyset$, it holds that 
    \begin{align}
        \beta_{\tilde{q}} = u_i(M^*_i) > u_{i_1}(M^*_{i_1}) = \beta_{\tilde{q}}-1.
        \label{eq:phi equiv 1}
    \end{align}
    \begin{itemize}
        \item Case $i_1 < i$. We have $u_{i_1}(A'_{i_1}) = u_{i_1}(A^*_{i_1})$ by the choice of $i$.
        Moreover, since $u_i(C'_i)=u_i(C^*_i)=0$, we have $u_i(A'_i) = u_i(M'_i)<u_i(M^*_i) = u_i(A^*_i) = \beta_{\tilde{q}}$ by~\eqref{eq:phi equiv 1}.
        This together with the assumption $u_{i_1}(A'_{i_1}) < u_i(A'_i)+1$ implies that 
        $$u_i(A'_i) \leq \beta_{\tilde{q}} - 1 < u_{i_1}(A^*_{i_1}) = u_{i_1}(A'_{i_1}) < u_i(A'_i)+1 \leq \beta_{\tilde{q}}.$$
        Thus, $u_i(A'_i) < u_{i_1}(A'_{i_1})$ and $u_{i_1}(C'_{i_1}) >0$ must hold.
        However, this contradicts the water-filling property of $\cA'$.

        \item Case $i_1 > i$. It holds that $u_i(A'_i) \leq u_{i_1}(A'_{i_1}) < u_i(A'_i)+1$.
        By the water-filling property of $\cA'$ and $u_i(C'_i)=0$, it also holds that $u_{i_1}(C'_{i_1})=0$ or $u_i(A'_i) = u_{i_1}(A'_{i_1})$.
        When $u_{i_1}(C'_{i_1})=0$, it turns out $u_i(A'_i) = u_{i_1}(A'_{i_1})$ because $u_i(A'_i)$ is integral.
        Then by~\eqref{eq:phi equiv 1}, we see that $i_1$ must be smaller than $i$ by the choice of ordering, which is a contradiction.
    \end{itemize}
    Thus, the only possible case is when $u_{i_1}(C^*_{i_1})=0$. 
    \begin{itemize}
        \item Case $i_1 < i$. We have $u_{i_1}(A'_{i_1}) = u_{i_1}(A^*_{i_1})~(=u_{i_1}(M^*_{i_1}))$ by the choice of $i$.

        \item Case $i_1 > i$. We have $u_i(A'_i) \leq u_{i_1}(A'_{i_1}) < u_i(A'_i)+1$.
        By the water-filling property of $\cA'$ and $u_i(C'_i)=0$, it also holds that $u_i(A'_i) = u_{i_1}(A'_{i_1})$.
        If additionally $u_{i_1}(A'_{i_1})> u_{i_1}(A^*_{i_1})$, then 
        $u_{i_1}(A^*_{i_1}) < u_{i_1}(A'_{i_1})=u_i(A'_i) \leq u_{\ell_i}(A^*_{\ell_i})$ since $i \leq k$.
        Here, recall that $u_{\ell_h}(A^*_{\ell_h}) = u_h (A^*_h)$ for all $h\in [i-1]$.
        This implies that the $i$th smallest value in $z(\cA^*)$ is less than $u_{\ell_i}(A^*_{\ell_i})$, which contradicts the ordering.
        Thus, $u_{i_1}(A'_{i_1}) \leq u_{i_1}(A^*_{i_1})=u_{i_1}(M^*_{i_1})$.
    \end{itemize}

    Therefore, in both Cases 2-3-1 and 2-3-2, we conclude that $u_{i_1}(M'_{i_1}) \leq u_{i_1}(M^*_{i_1})$.
    Since $g_1 \notin M^*_{i_1}$, there exists another indivisible good $g_2 \in M^*_{i_1}\setminus M'_{i_1}$.
    Let $i_2$ be the agent who has $g_2$ under $\cA'$, i.e., $g_2\in M'_{i_2}$.

    Now we repeat the same argument in Case 2 by replacing $i_1$ with $i_2$ as follows.
    \begin{itemize}
        \item If Case 2-1 applies (i.e., $u_{i_2}(A'_{i_2}) > u_i(A'_i)+1$), then we transfer $g_2$ to agent $i_1$ and $g_1$ to agent $i$, and the resulting allocation satisfies the second statement of the lemma.

        \item In Case 2-2, we can derive a contradiction in a similar way by transferring $g_2$ to agent $i_1$ and $g_1$ to agent $i$.

        \item In Case 2-3, we assume that $u_{i_2}(A'_{i_2}) < u_i(A'_i)+1$.
    Because $g_2 \in M'_{i_2} \cap M^*_{i_1}$, we observe that $i_2 \in N^1\cup \dots \cup N^{s(i_1)} \subseteq N_1\cup \dots \cup N^{s(i)}$ by~\eqref{eq:phi equiv 2}.
    We use this instead of~\eqref{eq:phi equiv 2}.
    In Case 2-3-1 ($u_i(C'_i)> 0$), if $u_i(C'_i) \geq 1$, or $u_i(C'_i) < 1$ and $u_{i_2}(A'_{i_2})$ is integral, then we can derive a contradiction in the same way as Case 1-1 by transferring $g_2$ to agent $i_1$, $g_1$ to agent $i$, and an appropriate part of agent $i$'s piece to agent $i_2$.
    Following the remaining discussions by just replacing $i_1$ with $i_2$, we conclude that $u_{i_2}(M'_{i_2}) \leq u_{i_2}(M^*_{i_2})$.
    \end{itemize}
        
    After we repeat the argument in Case 2-3 $t$ times, we obtain 
    a sequence $i_0=i, g_1, i_1, g_2, \dots, g_t, i_t$ of goods and agents satisfying $g_h \in (M^*_{i_{h-1}} \setminus M'_{i_{h-1}}) \cap M'_{i_h}$ for $h=1\dots, t$.
    If the same agent appears in the sequence, i.e., $i_b= i_{b'}$ with $b<b'\leq t$, then we can decrease the value of~\eqref{eq:symmetric} without changing the utility vector by transferring the indivisible good $g_{h+1}$ to agent $i_{h}$ for $h=b,b+1\dots,b'-1$, which contradicts the choice of $\cA'$.
    Since the number of agents is finite, the sequence will terminate at some agent $i_t$ with $u_{i_t}(A'_{i_t}) > u_i(A'_i) +1$ (Case 2-1).
    By transferring the good $g_h$ to agent $i_{h-1}$ for $h=1,\dots, t$, we obtain a utilitarian optimal allocation satisfying the second statement of the lemma.
\end{proof}

Now we are ready to prove Theorem~\ref{thm:phi-equiv mixed goods}.
\begin{proof}[Proof of Theorem~\ref{thm:phi-equiv mixed goods}]
    We assume $C\neq \emptyset$ since otherwise what we show is exactly Theorem~\ref{thm:Phi-either type}.
    Let $\Phi$ be any symmetric strictly convex function.

    \paragraph{1$\Rightarrow$2}
    Suppose that $\cA$ is $\Phi$-fair but not leximin.
    Let $x \coloneqq z(\cA)^\uparrow$.
    Then by Lemma~\ref{lem:decrease phi}, there exists a utilitarian optimal allocation $\hat{\cA}$ satisfying either statement in the lemma.
    If $z(\hat{\cA})^\uparrow=x+\varepsilon(\chi_i-\chi_j)$ for some $\varepsilon \in (0,1)$ and $i,j\in [n]$ such that $x_j > x_i + \varepsilon$,
    then in a similar way to~\eqref{eq:decrease Phi}, we see that
    \begin{align*}
        \Phi(z(\cA)) &= \Phi(x) \\
        &= \frac{\varepsilon}{x_j-x_i}\Phi(x - (x_j-x_i)(\chi_j-\chi_i))+(1-\frac{\varepsilon}{x_j-x_i})\Phi(x)\\
        &>\Phi\left(\frac{\varepsilon}{x_j-x_i}(x - (x_j-x_i)(\chi_j-\chi_i))+ (1-\frac{\varepsilon}{x_j-x_i})x\right)\\
        &=\Phi(x+\varepsilon(\chi_i-\chi_j)) = \Phi(z(\hat{\cA})).
    \end{align*}
    We have $\Phi(z(\cA))>\Phi(z(\hat{\cA}))$ also in the other case.
    Thus, we have a contradiction to $\Phi$-fairness of $\cA$, and hence $\cA$ is leximin.

    \paragraph{2$\Rightarrow$1}
    By the proof of 1$\Rightarrow$2, any $\Phi$-fair allocation is leximin.
    Since the utility vector of an leximin allocation is unique up to value-equivalence, any leximin allocation has the same utility vector as a $\Phi$-fair allocation up to value-equivalence.
    Since the function value of $\Phi$ depends only on the utility vector and $\Phi$ is symmetric, any leximin allocation is also $\Phi$-fair.

    \paragraph{2$\iff$3}
    Let $\Phi'(z)=-\sum_{i\in N}\log (z_i)$.
    Since we assume that $C\neq \emptyset$ and the valuations over $C$ is identical, all the agents have positive utility in MNW allocations.
    If $\cA$ is not utilitarian optimal, then there exist an indivisible good $g$ such that $g$ is given to an agent who values it as $0$.
    This implies that $\cA$ is not an MNW allocation.
    Thus, $\cA$ maximizes the Nash welfare if and only if $\cA$ is $\Phi'$-fair.
    Because we have proven 1$\iff$2 for any $\Phi$, 
    we see that $\cA$ is leximin if and only if $\cA$ is $\Phi'$-fair.
    Therefore, the equivalence 2$\iff$3 holds.
\end{proof}

\end{document}